\documentclass[11pt]{amsart}
\usepackage{amsmath,amstext,amsthm,amsfonts,amssymb}
\usepackage{multicol}
\usepackage{latexsym}
\usepackage{color}
\usepackage{graphicx}
\usepackage{epsf}
\usepackage{epsfig}
\usepackage{subfigure}
\usepackage{rotating}
\usepackage{curves}
\usepackage{epic}

\usepackage{graphics}
\usepackage{verbatim}
\usepackage{color}
\usepackage{hyperref}

\newtheorem{theorem}{Theorem}[section]
\newtheorem{lemma}[theorem]{Lemma}
\theoremstyle{definition}
\newtheorem{definition}[theorem]{Definition}
\newtheorem{proposition}[theorem]{Proposition}

\theoremstyle{remark}
\newtheorem{remark}[theorem]{Remark}
\numberwithin{equation}{section}

\newcommand{\HI}{\mathfrak{H}}

\newcommand{\C}{\mathbb{C}}

\newcommand{\B}{\mathcal{B}}

\usepackage{amssymb}
\usepackage{amstext}
\usepackage{amsmath}
\usepackage{latexsym}
\usepackage{color}

\newcommand{\quat}{\mathbb H}

\newcommand{\mc}{\mathcal}

\newcommand{\be}{\begin{equation}}
\newcommand{\en}{\end{equation}}
\newcommand{\D}{{\mc D}}

\newcommand{\N}{\mathbb N}
\newcommand{\Hil}{\mc H}

\newcommand{\bedefin}{\begin{defi}}
\newcommand{\findefi}{\end{defi} \medskip}

\newcommand{\betheo}{\begin{theorem}$\!\!${\bf \,\,\,}}
\newcommand{\entheo}{\end{theorem}}
\newcommand{\enth}{\end{theorem}}

\newcommand{\becor}{\begin{cor}$\!\!${\bf .}}
\newcommand{\encor}{\end{cor}}

\newcommand{\belem}{\begin{lem}$\!\!${\bf .}}
\newcommand{\enlem}{\end{lem}}

\newcommand{\bea}{\begin{eqnarray}}
\newcommand{\ena}{\end{eqnarray}}

\newcommand{\beano}{\begin{eqnarray*}}
\newcommand{\enano}{\end{eqnarray*}}

\newcommand{\bee}{\begin{enumerate}}
\newcommand{\ene}{\end{enumerate}}

\newcommand{\bei}{\begin{itemize}}
\newcommand{\eni}{\end{itemize}}

\newcommand{\betab}{\begin{tabular}}
\newcommand{\entab}{\end{tabular}}

\newcommand{\bd}{\begin{displaymath}}

\newcommand{\h}{{\mathfrak H}}

\newcommand{\hk}{{\mathfrak H}_{K}}

\newcommand{\bx}{\mathbf x}

\newcommand{\fk}{\mathfrak{K}}

\newcommand{\oqu}{\overline{q}}

\usepackage{subfigure}

\begin{document}

\title[Quaternionic reproducing kernels ]{ General construction of Reproducing Kernels on a quaternionic Hilbert space }
\author{K. Thirulogasanthar$^{1}$}
\address{$^{1}$ Department of Computer Science and Software
Engineering, Concordia University, 1455 De Maisonneuve Blvd. West,
Montr\'eal, Qu\'ebec, H3G 1M8, Canada. }
\email{santhar@gmail.com}
\author{S. Twareque Ali$^{2}$}
\address{$^2$Department of Mathematics and
Statistics, Concordia University,
Montr\'eal, Qu\'ebec, H3G 1M8, Canada.}
\email{twareque.ali@concordia.ca}

\thanks{The research of S. Twareque Ali was partly supported by the Natural Sciences and Engineering Research Council of Canada (NSERC) and  the research of K. Thirulogasanthar was partly supported by The Fonds de recherche du Qu\'ebec - Nature et technologies (FRQNT)}
\date{\today}
\begin{abstract}
A general theory of reproducing kernels and reproducing kernel Hilbert spaces on a right quaternionic Hilbert space is presented. Positive operator valued measures and their connection to a class of generalized quaternionic coherent states are examined. A Naimark type extension theorem associated with the positive operator valued measures is proved in a right quaternionic Hilbert space. As illustrative examples, real, complex and quaternionic reproducing kernels and reproducing kernel Hilbert spaces arising from Hermite and Laguerre polynomials are presented. In particular, in the Laguerre case, the Naimark type extension theorem on the associated quaternionic Hilbert space is indicated.
\end{abstract}
\maketitle
\pagestyle{myheadings}
\section{Introduction}

In this paper we study  the problem of defining reproducing kernels on quaternionic Hilbert spaces in a slightly different method from the existing theory. In fact, we present the theory in a physicist point of view.
It is needless to point out the usefulness of reproducing kernels, both in mathematics and physics. The subject has a long history both in the  mathematical and physical literatures (see for example, \cite{Alibk, Aron,Mesh}).
As is well known, for Hilbert spaces over the complex numbers, the existence of a reproducing kernel is basic to defining families of coherent states and frames (see, for example \cite{alirev,Alibk,AliGor}). Here we try to define reproducing kernels on quaternionic Hilbert spaces in basically the same way as on complex Hilbert spaces, following the construction given in \cite{szaf,sz1,sz2}. It turns out that the construction goes through almost verbatim as in the complex case. We then use the kernel to define coherent states and positive operator valued measures and also give some examples which are not available in the literature of quaternionic reproducing kernels.\\

Due to the non-commutativity of the quaternions there are three types of quaternionic Hilbert spaces, the left quaternionic Hilbert space, the right quaternionic Hilbert space and two-sided Hilbert space, depending on how multiplication of vectors by quaternions is defined. Various aspects of reproducing kernel spaces with an indefinite metric in the quaternionic setting is presented in \cite{Alpay}. In particular, reproducing kernels were discussed in the more general setting of Pontryagin quaternionic spaces (see also \cite{al-col} ). In this paper, we take a different approach. We start from a basis to construct reproducing kernels in a quaternionic Hilbert space and prove various results in physical point of view. It may also be of relevance to note that using the single indexed and doubly indexed quaternionic Hermite polynomials a particular class of reproducing kernels and reproducing kernel Hilbert spaces and their associated coherent states have been obtained in \cite{Thi1}. In the quaternionic setting, connections between reproducing kernels and continuous frames have also been discussed in \cite{Kho}.

The article is organized as follows. In section 2, for the sake of completeness, we briefly revisit the theory of complex reproducing kernels. In Section 3 we gather some facts about quaternions, quaternionic Hilbert spaces and quaternionic functional analysis, as needed for the development of the manuscript. We also prove a few facts that are not currently available in the literature. In Section 4.1 we develop a general theory of reproducing kernels and reproducing kernel Hilbert spaces on a right quaternionic Hilbert space, following the procedure used for the case of  complex Hilbert spaces as presented in \cite{Alibk, sz1, sz2}. Using the reproducing kernels, a class of generalized coherent states (CS), positive operator valued (POV) measures and an associated Naimark type extension theorem are studied in Sections 4.2 and 4.3. A reproducing kernel associated with the quaternionic canonical CS is presented in Section 4.4. As illustrative examples, real, complex and quaternionic reproducing kernels and reproducing kernel Hilbert spaces arising from Hermite and generalized Laguerre polynomials are presented in Sections 5.1 and 5.2 respectively. Section 6 ends the manuscript with some conclusions.

\section{The complex reproducing kernels}
Before introducing reproducing kernels on a quaternionic Hilbert spaces, for the sake of completeness, let us introduce it on a complex Hilbert space. However, we shall not go deep into details because the complex theory is very well documented, for example, in \cite{Alibk}. The material of this section is extracted from \cite{Alibk} and \cite{sz2}.\\

Let $X$ be a locally compact set and $\fk_c$ (we reserve the symbol $\fk$ for the quaternion case) a finite dimensional complex Hilbert space of dimension $n$. Let $f_i: X\longrightarrow \fk_c$, $i=1,2,\cdots, N,$ where $N$ could be finite or infinite, be a set of vector valued functions which satisfy the following conditions.
\begin{enumerate}
\item[(a)] For each $x\in X,$
\begin{equation}\label{k1}
0<\sum_{i=1}^N\|f_i(x)\|_{\fk_c}^2<\infty,
\end{equation}
where $\|\cdot\|_{\fk_c}$ denotes the norm in $\fk_c$.
\item[(b)]  For any sequence of complex numbers $c_0, c_1,\cdots c_N,$ such that $\displaystyle \sum_{i=0}^N |c_i|^2<\infty$,
\begin{equation}\label{k2}
\sum_{i=0}^N c_if_i(x)=0\quad\text{for all} \quad x\in X\Leftrightarrow c_i=0 \quad\text{for all}\quad i=0,1,\cdots N.
\end{equation}
\item[(c)] For each $x\in X$ the set of vectors $f_i(x),\quad i=0,1,\cdots, N,$ spans $\fk_c$.
\end{enumerate}
Define a positive definite kernel $K: X\times X\longrightarrow \mathcal{L}(\fk_c)$,
\begin{equation}\label{k3}
K(x,y)=\sum_{i=0}^{\infty}|f_i(x)\rangle\langle f_i(y)|
\end{equation}
the convergence of the sum being guaranteed by \ref{k1}. The positive definiteness of $K$ is understood in the following sense.
\begin{enumerate}
\item [(i)] if $x_1, x_2,\cdots x_m$ is any finite set of points in $X$ and $v_i\in\fk_c, i=1,2,\cdots, m, $ are arbitrary vectors, then
\begin{equation}\label{k4}
\sum_{i,j=1}^m \langle v_i | K(x_i, x_j) v_j\rangle\geq 0.
\end{equation}
\item[(ii)] For each $x\in X$, the operator $K(x,x)\in\mathcal{L}(\fk_c) $ is strictly positive definite, that is,
\begin{equation}\label{k5}
K(x,x)v=0 \Leftrightarrow v=0.
\end{equation}
\end{enumerate}
Now it is straight forward that $K(x,y)=\overline{K(y,x)}$ for all $x,y \in X$. For each $x\in X$ define the function $\xi_x^v: X\longrightarrow\fk_c$ by
\begin{equation}\label{k6}
\xi_x^v(y)=\sum_{i=0}^N\langle f_i(x) | v\rangle f_i(y)=K(y,x)v.
\end{equation}
Using this definition we can also write
\begin{equation}\label{k7}
\xi_x^v=K(\cdot, x)v=\sum_{i=o}^{N}\langle f_i(x) | v\rangle f_i.
\end{equation}
Let $\mathfrak{H}_c$  (the symbol $\HI$ is reserved for the case of quaternions) be the vector space formed by taking the complex linear span of these functions. On $\HI_c$ define a scalar product $\langle\cdot |\cdot\rangle_K$ by
\begin{equation}\label{k8}
\langle\xi_x^u | \xi_y^v\rangle_K=\langle u | K(x,y)v\rangle_{ \fk_c}
\end{equation}
this defines a non-degenerate sesquilinear form and hence a scalar product on $\HI_c$. Take the completion of $\HI_c$ with respect to the norm defined by $\langle\cdot | \cdot\rangle_K$ to get a Hilbert space $\HI_K^c$ ( the symbol $\HI_K$ is reserved for quaternion case).\\

If the dimension of $\fk_c$ is one, $K(x,y)$ is a complex valued function and the vectors in \ref{k7} are $\xi_x=K(\cdot, x),$ $x\in X$.\\

From \ref{k6} and \ref{k8} we have
\begin{equation}\label{k9}
\langle\xi_x^u | \xi_y^v\rangle_K=\langle u | \xi_y^v(x)\rangle_{\fk_c}
\end{equation}
and since an arbitrary element $\Phi\in\HI_K^c$ can be written as a linear combination of the vectors $\xi_y^v$, we get the reproducing property of the kernel:
\begin{equation}\label{k10}
\langle K(\cdot, x)u |\Phi\rangle_K=\langle \xi_x^u | \Phi\rangle_K=\langle u | \Phi(x)\rangle_{\fk_c}.
\end{equation}

If $\fk_c$ is one dimensional \ref{k10} takes the familiar form
$$\langle K(\cdot , x) | \Phi\rangle_K=\langle \xi_x | \Phi\rangle_K=\Phi(x).$$

Let $v^i, i=1,2,\cdots, n $ be an orthonormal basis of $\fk_c$ and set $\xi_x^i=\xi_x^{v^i}.$ The Hilbert space $\HI_K^c$ is called a reproducing kernel Hilbert space with reproducing kernel $K(x, y)$ and the vectors $\{\xi_x^i | x\in X, i=1,2,\cdots, n\}$ its associated coherent states. For an enhanced explanation along these lines of argument we refer the reader to \cite{Alibk}.
\section{Mathematical preliminaries : Quaternions}
For completeness we recall few facts about quaternions, quaternionic Hilbert spaces and the quaternionic functional calculus which may not be very familiar to the reader. For quaternions and quaternionic Hilbert spaces we refer the reader to \cite{Ad}. For quaternionic functional analysis we refer to \cite{Alpay2} and to the review article \cite{Ric} and references therein.
\subsection{Quaternions}
Let $\mathbb{H}$ denote the skew field of quaternions. Its elements are of the form $q=x_0+x_1i+x_2j+x_3k,~$ where $x_0,x_1,x_2$ and $x_3$ are real numbers, and $i,j,k$ are imaginary units such that $i^2=j^2=k^2=-1$, $ij=-ji=k$, $jk=-kj=i$ and $ki=-ik=j$. The quaternionic conjugate of $q$ is defined to be $\overline{q} = x_0 - x_1i - x_2j - x_3k$. Quaternions do not commute in general. However $q$ and $\oqu$ commute, and quaternions commute with real numbers. Also, $|q|^2=q\oqu=\oqu q$ and $\overline{qp}=\overline{p}~\oqu.$ Considered as $\mathbb R^4$, the quaternionic field is measurable and we shall introduce convenient measures, $d\mu$,  on it. For instance $d\mu$ can be taken as a Radon measure or $d\mu=d\lambda d\omega$, where $d\lambda$ is a Lebesgue measure on $\C$ and $d\omega$  is a Harr measure on $SU(2)$. For details we refer the reader to, for example, \cite{Thi1} (page 12).

\subsection{Right Quaternionic Hilbert Space}
Let $V_{\mathbb{H}}^{R}$ be a linear vector space under right multiplication by quaternionic scalars (again $\mathbb{H}$ standing for the field of quaternions).  For $\phi ,\psi ,\omega\in V_{\mathbb{H}}^{R}$ and $q\in \mathbb{H}$, the inner product
$$\langle\cdot\mid\cdot\rangle:V_{\mathbb{H}}^{R}\times V_{\mathbb{H}}^{R}\longrightarrow \mathbb{H}$$
satisfies the following properties
\begin{enumerate}
\item[(i)]
$\overline{\langle\phi \mid \psi \rangle}=\langle \psi \mid\phi \rangle$
\item[(ii)]
$\|\phi\|^{2}=\langle\phi \mid\phi \rangle>0$ unless $\phi =0$, a real norm
\item[(iii)]
$\langle\phi \mid \psi +\omega\rangle=\langle\phi \mid \psi \rangle+\langle\phi \mid \omega\rangle$
\item[(iv)]
$\langle\phi \mid \psi q\rangle=\langle\phi \mid \psi \rangle q$
\item[(v)]
$\langle\phi q\mid \psi \rangle=\overline{q}\langle\phi \mid \psi \rangle$
\end{enumerate}
where $\overline{q}$ stands for the quaternionic conjugate. Conditions (i), (iii) -- (v)
define the sesquilinearity of the scalar product, while (ii) expresses its non-degeneracy.
We assume that the
space $V_{\mathbb{H}}^{R}$ is complete under the norm given above. Then,  together with $\langle\cdot\mid\cdot\rangle$ this defines a right quaternionic Hilbert space, which we shall assume to be separable. Quaternionic Hilbert spaces share most of the standard properties of complex Hilbert spaces. In particular, the Cauchy-Schwartz inequality holds on quaternionic Hilbert spaces as well as the Riesz representation theorem for their duals. Thus, the Dirac bra-ket notation
can be adapted to right quaternionic Hilbert spaces, with $\vert\phi \rangle$ denoting the vector $\phi$ and $\langle\phi \vert$ its dual vector. Moreover, for Hilbert spaces of quaternion valued functions $\phi (x)$, of some variable $x$, we also have
$$(\mid\phi q\rangle) (x)  =(\mid\phi \rangle (x) )q ,\hspace{1cm}(\langle\phi q\mid) (x)  =\overline{q}(\langle\phi \mid (x))\;, $$

Left quaternionic Hilbert spaces, with quaternionic scalar multiplication being defined from the left, may be analogously defined. However, here we shall be using only right quaternionic Hilbert spaces.

Once an orthonormal basis $\vert \phi_i\rangle, \; i \in J$, is chosen in a right quaternionic Hilbert space, a left multiplication by a quaternionic scalar $q$ may also be defined, albeit in a rather non-canonical fashion (see, for example, \cite{Vis,Ric}):
$$ q \vert \phi \rangle  = \sum_{i \in J} \vert \phi_i \rangle q \langle \phi_i\mid \phi\rangle . $$
Consequently, if the Hilbert space consists of quaternion valued functions, we have the rather unconventional result that
$$ (q \vert \phi \rangle )(x) \neq q (\vert \phi \rangle (x)) ). $$

A linear operator $A$, on a right quaternionic Hilbert space will be assumed to act from the left, $ (A,\; \vert \phi \rangle ) \longmapsto A\vert \phi \rangle$ and a scalar multiple $qA, \, q \in \mathbb H$ of such an operator will be defined in the usual way \cite{Ric},
$$ (qA)\vert \phi \rangle  = q (A\vert \phi \rangle). $$
However, many of the usual properties of scalar multiples, that hold on complex Hilbert
spaces, may not hold for left scalar multiples of operators on right quaternionic Hilbert
spaces defined in the above manner \cite{Ric,Mu,Sha}.

\begin{theorem} Let $\phi,\psi\in V_{\quat}^R$, then
\begin{enumerate}
\item[(a)]the Cauchy-Schwarz inequality holds: $|\langle \phi|\psi\rangle|^2\leq\langle\phi|\phi\rangle\langle\psi|\psi\rangle.$
\item[(b)]a polarization identity holds in the following form:
\begin{eqnarray*}
4\langle\phi|\psi\rangle=\|\phi+\psi\|^2&-&\|\phi-\psi\|^2
+(\|\phi i+\psi\|^2-\|\phi i-\psi\|^2)i\\
&+&(\|\phi j+\psi\|^2-\|\phi j-\psi\|^2)j
+(\|\phi k+\psi\|^2-\|\phi k-\psi\|^2)k.
\end{eqnarray*}
\end{enumerate}
\end{theorem}
Let
$$\h_{rd}=\{h:V_{\quat}^R\longrightarrow\quat~|~h~~ \mbox{is bounded and right linear}\}$$
be the right dual space of $V_\quat^R$ (it can also be turned into a left quaternionic Hilbert space \cite{Ale, Thi1}). With the usual norm
$$\|h\|=\sup\{|h(\phi)|~|~\|\phi\|=1, \phi\in V_\quat^R\}$$
$\h_{rd}$ is a real vector space and it can be transformed into a quaternionic Hilbert space \cite{Thi1}.
\begin{theorem}(Riesz) For any functional $h\in\h_{rd}$
$$h(\phi)=\langle\psi|\phi\rangle;\quad\phi\in V_\quat^R$$
for a vector $\psi\in V_\quat^R$, and then $\|h\|=\|\psi\|.$
\end{theorem}
Let $A$ be a densely defined right linear operator on $V_\quat^R$. Its adjoint  $A^{\dagger}$ and the domain of the adjoint are defined in the same way as for operators on complex Hilbert spaces. In particular,
an operator $A$ is said to be self-adjoint if $A=A^{\dagger}$. If $\phi \in V_\mathbb{H}^R\smallsetminus\{0\}$, then $|\phi \rangle\langle\phi |$ is a rank one projection operator. For operators $A, B$, by convention, we have
\begin{equation}\label{Rank}
|A\phi\rangle\langle B\phi |=A|\phi \rangle\langle\phi |B^{\dagger}.
\end{equation}
Let $\D(A)$ denote the domain of $A$. $A$ is said to be right linear if
$$A(\phi q+\psi p)=(A\phi )q+(A\psi )p;\quad\forall\phi ,\psi \in\D(A), q, p\in \mathbb{H}.$$
The set of all right linear operators will be denoted by $\mathcal{L}(V_\mathbb{H}^R)$. For a given $A\in \mathcal{L}(V_\mathbb{H}^R)$, the range and the kernel will be
\begin{eqnarray*}
\mbox{ran}(A)&=&\{\psi \in V_\mathbb{H}^R~|~A\phi =\psi \quad\text{for}~~\phi \in\D(A)\}\\
\ker(A)&=&\{\phi \in\D(A)~|~A\phi =0\}.
\end{eqnarray*}
We call an operator $A\in \mathcal{L}(V_\mathbb{H}^R)$ bounded if
\begin{equation*}
\|A\|=\sup_{\|\phi \|=1}\|A\phi \|<\infty.
\end{equation*}
or equivalently, if there exists $K\geq 0$ such that $\|A\phi \|\leq K\|\phi \|$ for $\phi \in\D(A)$. The set of all bounded right linear operators will be denoted by $\B(V_\mathbb{H}^R)$.
The following definition is the same as for operators on a complex Hilbert space\cite{Hel}.
\begin{proposition}\label{SAD}
Let $A\in\B(V_\mathbb{H}^R)$. Then $A$ is self-adjoint if and only if for each $\phi\in V_{\mathbb{H}}^{R}$, $\langle A\phi\mid\phi\rangle\in\mathbb{R}$.
\end{proposition}
\begin{proof}
A proof is given in \cite{Fa} for a left quaternionic Hilbert space $V_{\mathbb{H}}^{L}$, and it can be easily adapted to  $V_{\mathbb{H}}^{R}$. For a proof for the sufficiency part one may also see Proposition 2.17 (b) in \cite{Ric}.
\end{proof}
\begin{definition}\label{Def1}
Let $A$ and $B$ be self-adjoint operators on $V_\mathbb{H}^R$. Then $A\leq B$ ($A$ less or equal to $B$) or equivalently $B\geq A$ if $\langle A\phi|\phi\rangle\leq\langle B\phi|\phi\rangle$ for all $\phi\in V_\mathbb{H}^R$. In particular $A$ is called positive if $A\geq 0.$
\end{definition}
\begin{lemma}\label{expo}
Let $A\in\B(V_\mathbb{H}^R)$. If $A\geq 0$ and self-adjoint, then $A^n\geq 0$ for all $n\in\N$.
\end{lemma}
\begin{proof}
If $n$ is even, then
$$\langle A^n\phi|\phi\rangle=\langle A^{\frac{n}{2}}\phi|A^{\frac{n}{2}}\phi\rangle\geq 0;\quad \phi\in V_{\quat}^R.$$
If $n$ is odd and $n>1$, then
$$\langle A^n\phi|\phi\rangle=\langle A(A^{\frac{n-1}{2}}\phi)|A^{\frac{n-1}{2}}\phi\rangle\geq 0;\quad \phi\in V_{\quat}^R.$$
\end{proof}
\begin{theorem}\cite{Ric}\label{IT1}
Let $A\in\B(V_\mathbb{H}^R)$. If $A\geq 0$ then there exists a unique operator in $\B(V_\mathbb{H}^R)$, indicated by $\sqrt{A}=A^{1/2}$ such that $\sqrt{A}\geq 0$ and $\sqrt{A}\sqrt{A}=A$.
\end{theorem}
 \begin{lemma}\label{L2}\cite{Kho}
		 Let $A\in\B(V_\mathbb{H}^R)$ be a self-adjoint operator, then
 \begin{equation}\label{E13}
 \left\|A\right\|=\sup_{\left\|\phi\right\|=1}\left|\left\langle \phi|A\phi \right\rangle\right|
 \end{equation} 		
\end{lemma}
\begin{lemma}\cite{Kho}\label{ine}
Let $A$ be a positive self-adjoint operator on a Hilbert space $V_{\quat}^R$. Then
$$\|A\phi\|^2=\langle A^2\phi|\phi\rangle\leq \|A\|\langle A\phi|\phi\rangle;\quad\phi\in V_{\quat}^R.$$
\end{lemma}
\begin{theorem}\label{domi}
Suppose $\{A_n\}_{n=1}^{\infty}$ is a sequence of positive self-adjoint operators on a Hilbert space $V_\quat^R$ such that $A_n\leq A_{n+1}\leq M I_{V_\quat^R},$ for all $n\geq 1$, where $M\in\mathbb{R}$ is a constant and $I_{V_\quat^R}$ is the identity operator on $V_\quat^R$. Then there exists a self-adjoint operator $A:V_{\quat}^R\longrightarrow V_\quat^R$ such that
$$A\phi=\lim_{n\rightarrow\infty}A_n\phi;\quad\forall\phi\in V_\quat^R.$$
In other words $A_n$ converges to $A$ strongly.
\end{theorem}
\begin{proof}
For any $\phi\in V_\quat^R$, the real sequence $\{\langle \phi|A_n\phi\rangle\}_{n=1}^{\infty}$ is non-decreasing and bounded above by $M\|\phi\|^2$, hence converges to a real number $F(\phi)$ (say). For $\phi,\psi\in V_\quat^R$, since $\langle \phi|A_n\psi\rangle=\langle\sqrt{A_n}\phi|\sqrt{A_n}\psi\rangle$, from the polarization identity we can write
\begin{eqnarray*}
\langle \phi|A_n\psi\rangle&=&\frac{1}{4}\{\langle \phi+\psi|A_n(\phi+\psi)\rangle-\langle \phi-\psi|A_n(\phi-\psi)\rangle\\
&+&(\langle \phi i+\psi|A_n(\phi i+\psi)\rangle-\langle \phi i-\psi|A_n(\phi i-\psi)\rangle)i\\
&+&(\langle \phi j+\psi|A_n(\phi j+\psi)\rangle-\langle \phi j-\psi|A_n(\phi j-\psi)\rangle)j\\
&+&(\langle \phi k+\psi|A_n(\phi k+\psi)\rangle-\langle \phi k-\psi|A_n(\phi k-\psi)\rangle)k\}.
\end{eqnarray*}
Therefore, the limit $\displaystyle G(\phi,\psi):=\lim_{n\rightarrow\infty}\langle \phi|A_n\psi\rangle$ exists. {By the Cauchy-Schwarz inequality we have $\langle \phi|A_1\phi\rangle\leq \|A_1\|\|\phi\|^2$. We also have
$$-\langle \phi|A_1\phi\rangle\leq \langle \phi|A_n\phi\rangle\leq M\|\phi\|^2;\quad\forall\phi\in V_\quat^R.$$
Therefore,
$$-\|A_1\|\|\phi\|^2\leq \langle \phi|A_n\phi\rangle\leq M\|\phi\|^2;\quad\forall\phi\in V_\quat^R.$$
Hence, by Lemma (\ref{ine}), we have $\|A_n\|\leq M':=\max\{M, \|A_1\|\}.$
Again by the Cauchy-Schwarz inequality we get
$$|G(\phi,\psi)|\leq M'\|\phi\|\|\psi\|;\quad\forall\phi,\psi\in V_\quat^R.$$
Hence, for fixed $\psi$, $\overline{G(\phi,\psi)}$ is a bounded right linear functional with respect to $\phi\in V_\quat^R.$ Therefore, by the Riesz representation theorem, there exists a unique $A\psi\in V_\quat^R$ such that $\langle A\psi|\phi\rangle=\overline{G(\phi,\psi)}.$ That is,
$$G(\phi,\psi)=\langle \phi|A\psi\rangle;\quad\forall\psi\in V_\quat^R.$$
Let $\phi,\psi,\xi\in V_\quat^R$ and $\alpha, \beta\in\quat$. Then
\begin{eqnarray*}
\langle\xi|A(\phi\alpha+\psi\beta)\rangle&=&G(\xi,\phi\alpha+\psi\beta)
=\lim_{n\rightarrow\infty}\langle\xi|A_n(\phi\alpha+\psi\beta)\rangle\\
&=&\lim_{n\rightarrow\infty}\langle\xi|A_n\phi\alpha+A_n\psi\beta\rangle
=\lim_{n\rightarrow\infty}\langle\xi|A_n\phi\rangle\alpha+\langle\xi| A_n\psi\rangle\beta\\
&=&G(\xi,\phi)\alpha+G(\xi,\psi)\beta=\langle\xi|A\phi\rangle\alpha+\langle\xi| A\psi\rangle\beta\\
&=&\langle\xi|A\phi\alpha+A\psi\beta\rangle.
\end{eqnarray*}
That is, $A(\phi\alpha+\psi\beta)=A\phi\alpha+A\psi\beta$. Hence $A$ is linear. Since $\|A_n\|\leq M'$, $\|A\|\leq M'$ thereby $A$ is bounded. Since
$$\langle A\phi|\psi\rangle=\lim_{n\rightarrow\infty}\langle A_n\phi|\psi\rangle
=\lim_{n\rightarrow\infty}\langle \phi|A_n\psi\rangle=\langle \phi|A\psi\rangle;\quad\forall\phi,\psi\in V_\quat^R,$$
$A$ is self-adjoint. Now, by lemma (\ref{ine}), we have
\begin{eqnarray*}
\|(A-A_n)\phi\|^2=\langle(A-A_n)^2\phi|\phi\rangle
&\leq& \|A-A_n\|\langle(A-A_n)\phi|\phi\rangle\\
&\leq& 2M'\langle(A-A_n)\phi|\phi\rangle\rightarrow 0\quad\mbox{as}~~n\rightarrow\infty,
\end{eqnarray*}
which completes the proof.}
\end{proof}

\section{Reproducing kernels: general set up}\label{sec-rep-ker}
 On complex Hilbert spaces, reproducing kernels are typically defined on
Hilbert spaces of complex-valued functions. However, a more general setting, using
vector valued functions (with values in another, usually finite dimensional, Hilbert space)
have also been used in the literature (see, for example, \cite{alirev,Alibk} and references cited therein).
Here we adopt an analogous general setting in constructing quaternionic reproducing kernel
Hilbert spaces, which of course has spaces of quaternion valued functions as a special
case.

\subsection{The kernel and associated Hilbert space}\label{subsec-ker-hilb-sp}
As before, let $\quat$ be the field of quaternions, $X$ a locally compact space and
$\fk$  a right quaternionic Hilbert space of dimension $n < \infty$. Let $ f_i :X \
\longrightarrow\fk$ for $ i=1, 2, 3,\cdots, N$, where $N$ could be finite or infinite, be
a set of $\fk$-valued functions, on which we define multiplication by a quaternion $q$ from
the right in the obvious manner, $(f_i q)(x) = f_i (x)q$, and we assume these functions
to satisfy the following three conditions.
\begin{enumerate}
\item[(a)] For each $x\in X,$
\begin{equation}\label{n}
 \mathcal{N}(x)=\sum_{i=0}^N \|f_i(x)\|_{\fk}^2< \infty.
\end{equation}
\item[(b)] If $c_0, c_1, c_2, ...,c_N\in\quat$ is any sequence of quaternionic scalars satisfying
 $$
\sum_{i=0}^N|c_i|^2< \infty,$$ then
\begin{equation}\label{in-1}
\sum_{i=0}^N f_i(x)c_i=0 \ \ \mbox{for all } x \in X \  \ \ \mbox{if and only if } \ \
c_i=0, \; \text{for all} \;\; i.
\end{equation}
\item[(c)] For each $x\in X$,
\begin{equation}\label{fk}
\fk=\overline{\mbox{right-span}\{f_i(x)~\vert~i=0,1,\cdots,N\}}.
\end{equation}
\end{enumerate}
\medskip
Note that (c) also implies that $\mathcal N (x) > 0$, for all $x\in X$.
Using these functions we define a kernel
\begin{equation}\label{pre-K}
K: X \times X \longrightarrow \mathcal{L}(\fk)\ \ \mbox{by} \ \  K(x, y)= \sum_{i=0}^N|f_i(x)\rangle\langle f_i(y)| .
\end{equation}

Note that conditions $(a) - (c)$ are exactly the same conditions that are imposed while defining a reproducing kernel on a complex Hilbert space \cite{alirev}.

\begin{proposition}\label{P1}
The operator $K(x,x)$ is bounded and strictly positive.
\end{proposition}
\begin{proof}
Let $u\in\fk$ and consider
\begin{eqnarray*}
\|K(x,x)u\|_{\fk}&=&\|\sum_{i=0}^Nf_i(x)\langle f_i(x)|u\rangle_\fk\|_\fk\\
&\leq&\sum_{i=0}^N\|f_i(x)\langle f_i(x)|u\rangle_\fk\|_\fk\\
&=&\sum_{i=0}^N|\langle f_i(x)|u\rangle_\fk|~\|f_i(x)\|_{\fk}\\
&\leq&\sum_{i=0}^N\|f_i(x)\|^2_\fk\|u\|_\fk\quad\mbox{by Cauchy-Schwarz}\\
&=&\mathcal N (x)\|u\|_\fk < \infty,
\end{eqnarray*}
 by (a). Thus $K(x,x)$ is a bounded operator.
Let $u\in\fk$ and $u\not=0$, then
\begin{eqnarray*}
\langle u|K(x,x)u\rangle_{\fk}
&=&\sum_{i=0}^N\langle u|f_i(x)\rangle_{\fk}\langle f_i(x)|u\rangle_{\fk}\\
&=&\sum_{i=0}^N|\langle u|f_i(x)\rangle_\fk|^2>0,
\end{eqnarray*}
the last inequality following from condition (c) on the $f_i$.
Hence  $K(x,x)$ is strictly positive.
\end{proof}
Note that $K(x,y)=K(y,x)^*$ for all $x,y\in X$. If $\fk=\quat^n$,
then $K(x,y)$ is a $n\times n$ matrix with quaternionic entries and $K(x,x)$ has positive
non-zero eigenvalues.

Let
$v \in \mathfrak K$ and for each $x\in X$ define the functions
\begin{equation}
\xi_x^{v} :=\sum_{i=0}^N f_i\langle f_i(x)| v \rangle_\fk=K(\cdot,x)v,
\label{xi}
\end{equation}
i.e.,
$$
\xi_x^{v}(y)=\sum_{i=0}^N f_i(y)\langle f_i(x)| v \rangle_\fk=K(y,x)v,
$$
the finiteness of the norm of these vectors being guaranteed by the condition (a) on
the $f_i$.
Note that $\xi^v_x = 0 \Leftrightarrow v =0$. Also, if $q \in \mathbb H$, then $\xi_x^v q = \xi_x^{vq}$.

Let $$\Hil=\mbox{right-span}\{\xi_x^v~\vert~x\in X, v\in\fk\}.$$
A general element  $\phi \in \mathcal H$ consists of finite linear combinations of the type
\bea
\label{lin-comb}
\phi &=& \sum_{i,j} \xi^{v_j}_{x_i}c_{ij}, \quad c_{ij} \in \mathbb H\nonumber\\
     & = & \sum_i \xi^{u_i}_{x_i}, \quad u_i = \sum_j v_j c_{ij} .
\ena
On $\mathcal H$  define a sesquilinear form $\langle \cdot\mid\cdot\rangle_K$, by first setting
\begin{equation}\label{k-norm}
\langle\xi_x^u\vert\xi_y^v\rangle_K=\langle u\vert K(x,y)v\rangle_{\fk}=
              \langle u \mid \xi_y^v(x) \rangle_\fk
\end{equation}
and then extending by linearity, so that  for $\phi \in \mathcal H$,
\be
\label{lin-ext}
  \langle\xi_x^u\vert \phi\rangle_K = \langle u \mid \phi(x) \rangle_\fk.
 \en
It is clear that $\langle \xi^v_x \mid \xi^v_x \rangle_K =0 \Leftrightarrow \xi^v_x  =0$.
Moreover,
\be
\langle \xi^v_x \mid \xi^v_x \rangle_K \leq \Vert v\Vert^2_\fk\; \Vert K(x,x)\Vert.
\label{xi-norm}
\en

We now prove a stronger result, for which we first need a Lemma.
\begin{lemma}\label{pre-schwarz}
For  $\psi\in \mathcal H$ such that $\langle \psi \mid\psi\rangle \neq 0$ and arbitrary
$\phi \in \mathcal H$, we have the inequality
\be\label{soft-schwarz}
  \vert\langle \psi \mid \phi\rangle_K \vert \leq [\langle\psi \mid \psi
  \rangle_K]^{\frac 12}\;[\langle\phi \mid \phi \rangle_K]^{\frac 12}.
\en
\end{lemma}
\begin{proof}
First note that  $\langle\phi\mid\phi\rangle_K \geq 0$, for any $\phi \in \mathcal H$.
Indeed, from (\ref{lin-comb}),
\beano
 \langle\phi\mid\phi\rangle_K  & = &  \sum_{i,j}^n \langle \xi_{x_i}^{u_i}
 \mid \xi_{x_j}^{u_j}\rangle_K , \quad n<\infty \\
  & = & \sum_{i,j}^n\sum_{k=0}^N \langle u_i \mid K(x_i , x_j)u_j \rangle_{\mathfrak K},
  \quad \text {by} \;\; (\ref{k-norm})
\enano
Using (\ref{pre-K}) we get
$$ \langle\phi\mid\phi\rangle_K   =  \sum_{k=0}^N \vert q_k\vert^2 \geq 0, \quad q_k =
\sum_{j=0}^n \langle f_k (x_j )\mid u_j\rangle \in \mathbb H\; . $$

Next let $\psi \in \mathcal H$ be such that $\langle \psi \mid\psi\rangle
\neq 0$ and let $\phi \in \mathcal H$ be arbitrary. Consider the vector
$\psi q + \phi p, \; q, p \in \mathbb H$. Then
$\langle \psi q + \phi p\mid \psi q + \phi p \rangle_K \geq 0$. Expanding we get
$$ \vert q\vert^2 \langle\psi\mid\psi\rangle_K + \vert p \vert^2
\langle\phi\mid\phi\rangle_K +
 \overline{q}\langle\psi\mid\phi\rangle_K\; p + \overline{p}\langle\phi\mid\psi\rangle_K\; q
 \geq 0\; . $$
Taking $q = - \langle \psi \mid \phi\rangle$ and $p = \langle \psi \mid \psi\rangle$ in this inequality and dividing out by $p$ we get
$$ \vert\langle \psi\mid\phi\rangle_K\vert^2 \leq \langle\psi\mid\psi\rangle_K\; \langle\phi\mid\phi\rangle_K\; , $$
from which (\ref{soft-schwarz}) follows.
\end{proof}

\begin{proposition}\label{P2}
$\langle~\cdot~|~\cdot~\rangle_K$ defines a scalar product on $\Hil$.
\end{proposition}
\begin{proof}
Sesquilinearity is easy to check from the definition and the fact that $K(x,y)^* =
K(y,x)$. We only need to prove non-degeneracy. Let $\phi \in \mathcal H$ be a non-zero vector and suppose $\phi (x) \neq 0$. Then
\beano
\Vert\phi (x)\Vert^2_\fk &=& \langle \phi (x) \mid \phi (x)\rangle_\fk = \langle \xi^{\phi (x)}_x \vert \phi\rangle_K, \quad \text{by} \;\; (\ref{lin-ext})\\
&\leq & [\langle\xi_x^{\phi(x)} \mid \xi_x^{\phi (x)}
  \rangle_K]^{\frac 12}\;[\langle\phi \mid \phi \rangle_K]^{\frac 12}, \quad \text{by Lemma}\;\; \ref{pre-schwarz}\\
&=& [\langle\phi(x) \mid K(x,x)\phi (x)
  \rangle_\fk]^{\frac 12}\;[\langle\phi \mid \phi \rangle_K]^{\frac 12}\\
&\leq & \Vert\phi (x)\Vert\; \Vert K(x,x)\Vert^{\frac 12} \;[\langle\phi \mid \phi \rangle_K]^{\frac 12}.
\enano
Thus,
\be\label{norms-ineq}
\Vert\phi (x)\Vert_\fk \leq \Vert K(x,x)\Vert^{\frac 12} \;[\langle\phi \mid \phi \rangle_K]^{\frac 12},
\en
from which we see that $\langle\phi\mid\phi\rangle_K = 0 \; \Leftrightarrow\; \phi = 0$.
\end{proof}
For $\phi \in \mathcal H$, we define the norm
$\Vert \phi \Vert_K = [\langle \phi\mid\phi\rangle_K]^{\frac 12}$ and take the
right quaternionic completion of $\mathcal H$
in this norm to get a right quaternionic Hilbert space $\hk$.
We now show that $\hk$ is a reproducing kernel Hilbert space, with kernel $K(x,y)$.
\begin{proposition}\label{P3}
The kernel $K(x,y)$ is a reproducing kernel for  $\hk$.
\end{proposition}
\begin{proof}
Using (\ref{norms-ineq}) it is easy to extend (\ref{lin-ext}) to arbitrary $\phi \in \hk$. Thus,
for any $\phi \in \hk$ and $u \in \fk$,
\be
\langle K(\cdot , x)u \mid \phi \rangle_K = \langle \xi_x^u \mid \phi \rangle_K = \langle u\mid \phi (x)\rangle_\fk \; ,
\label{eval-map}
\en
which is the reproducing property of the kernel.
\end{proof}
In the case where $\fk = \mathbb H$, we just get,
$$
  \langle K(\cdot, x)\mid \phi\rangle_K = \phi (x), $$
which is the more standard form of the reproducing property. Just as in the case of reproducing kernels on complex Hilbert spaces, we have the following useful result.

\begin{proposition}\label{ker-unique}
Given the Hilbert space $\hk$, the reproducing kernel is unique. Moreover, if $g_i, \;\; i = 0, 1,2, \ldots, N$, is any orthonormal basis of $\hk$, then
\be
   K(x, y) = \sum_{i=0}^N \vert g_i (x) \rangle \langle g_i (y) \vert\; .
\label{ker-rep}
\en
\end{proposition}
The proof is entirely analogous to that for the complex case (see, for example, \cite{Aron,Mesh}) and we omit it here.

Let $\{v_i~|~i=i,2,\cdots,n\}$ be an orthonormal basis for $\fk$ and set $\xi_x^i:=\xi_x^{v_i}$.
\begin{proposition}\label{P4}
For each fixed $x\in X$, the set of vectors $\xi_x^i,~i=1,2,\cdots,n$, is linearly independent.
\end{proposition}
\begin{proof}
Indeed, let $q_i \in \mathbb H, \;\; i=1,2,3. \ldots , n$, and suppose $\phi:= \sum_{i=1}^n \xi_x^i q_i = 0$. But then
$$ \phi = \sum_{i=1}^n \xi_x^{v_i q_i} = \sum_{i=1}^n K(\cdot , x )v_iq_i = \xi_x^u , \quad \text{where,}\;\; u = \sum_{i=1}^n v_i q_i \; .$$
Thus, $\phi = 0 \; \Leftrightarrow \; u =0$ and since the $v_i$ form a basis, this
implies that $q_i =0$, for each $i$, i.e., the $\xi^i_x, \; i =1,2,3, \ldots , n$, are linearly independent.
\end{proof}
\begin{definition}\label{D1} Let $\{v_i~|~i=i,2,\cdots,n\}$ be an orthonormal basis for $\fk$ and set $\xi_x^i:=\xi_x^{v_i}$.
The Hilbert space $\hk$ is called a {\em (right) quaternionic reproducing kernel Hilbert space} (QRKHS)
with reproducing kernel $K(x,y)$. The set of vectors
\be
\mathfrak{G}_K=\{\xi_x^i~|~x\in X, i=1,2,\cdots,n\}
\label{qcsbasis}
\en
is called its associated set of (generalized) {\em quaternionic coherent states (QCS)}.
\end{definition}
Define the  operator-valued function  $F_K:X\longrightarrow \mathcal L (\hk)$ by
\begin{equation}\label{FK}
F_K(x)=\sum_{i=1}^n|\xi_x^i\rangle\langle\xi_x^i|
\end{equation}
\begin{proposition}\label{P5}
For each $x\in X$, the operator $F_K(x)$ is a rank-$n$,  positive and bounded operator.
\end{proposition}
\begin{proof}
Similar to that of Proposition (\ref{P1}).
\end{proof}
Note that the range $\mbox{ran}(F_K(x)):=\fk_x$ is an $n$-dimensional subspace of $\hk$. We say that the positive operator-valued (POV) function  $x\mapsto F_K(x)$  is canonically associated to the reproducing kernel Hilbert space $\hk$.
\begin{proposition}\label{P6}
Let $\phi, \psi\in\hk$, then $\langle\phi|F_K(x)\psi\rangle_K=\langle\phi(x)|\psi(x)\rangle_\fk$.
\end{proposition}
\begin{proof}
From Proposition (\ref{P3}) we have $\langle \phi|\xi_x^i\rangle_K=\langle\phi(x)|v_i\rangle_\fk$ and $\langle \xi_x^i|\psi\rangle_K=\langle v_i|\psi(x)\rangle_\fk$. Further, since $\{v_i|~i=1,2,\cdots n\}$ is an orthonormal basis of $\fk$, $\sum_{i=1}^n|v_i\rangle\langle v_i|=I_{\fk},$ the identity operator on $\fk$. Therefore, we have
\begin{eqnarray*}
\langle\phi|F_K(x)\psi\rangle_K&=&\langle\phi|\sum_{i=1}^n|\xi_x^i\rangle
\langle\xi_x^i|\psi\rangle_K
=\sum_{i=1}^n\langle\phi|\xi_x^i\rangle_K\langle\xi_x^i|\psi\rangle_K\\
&=&\sum_{i=1}^n\langle\phi(x)|v_i\rangle_\fk\langle v_i|\psi(x)\rangle_\fk
=\langle\phi(x)|\psi(x)\rangle_\fk.
\end{eqnarray*}
\end{proof}
In view of the above proposition, and in analogy with physical convention for a similar
operator function on a reproducing kernel Hilbert space over the complexes,
the operator $F_K(x)$ may be called the {\em localization operator} at the point $x$.

For each $x\in X$, define the map
\begin{equation}\label{eval}
E_K(x):\hk\longrightarrow\fk_x=\fk\quad\mbox{by}\quad E_K(x)\phi=\phi(x).
\end{equation}
Clearly, this map is linear and will be called the {\em evaluation map} at the point $x\in X$.

\begin{proposition}\label{P8}
For each $x\in X$, the evaluation map $E_K(x)$ is bounded.
\end{proposition}
\begin{proof}
The proof follows immediately by combining (\ref{eval}) with (\ref{norms-ineq}).
\end{proof}
\begin{proposition}\label{P10}
The vectors $\{f_i\}_{i=0}^N$ in (\ref{n}) satisfy $\langle f_i|f_j\rangle_K=\delta_{ij};~~i,j=0,1,\cdots N$.
\end{proposition}
\begin{proof}
From (\ref{k-norm}) and (\ref{xi}) we have
$$\langle\sum_{i=0}^Nf_i\langle f_i(x)|u\rangle_\fk|\sum_{j=0}^Nf_j\langle f_j(y)|u\rangle_\fk\rangle_K=\langle u|K(x,y)v\rangle_\fk.$$
That is
\begin{eqnarray}\label{del}
\sum_{i=0}^N\sum_{j=0}^N\langle u|f_i(x)\rangle_\fk\langle f_i|f_j\rangle_K\langle f_j(y)|v\rangle_\fk&=&\langle u|K(x,y)v\rangle_\fk\\
&=&\sum_{i=0}^N\langle u|f_i(x)\rangle_\fk\langle f_j(y)|v\rangle_\fk\nonumber
\end{eqnarray}
Therefore, comparing (\ref{del}) with (\ref{in-1}) and (\ref{pre-K}) we get
$$\langle f_i|f_j\rangle_K=\delta_{ij};\quad i,j=0,1,\cdots N.$$
\end{proof}

\begin{remark}\label{Re1}
From the above proposition, the set of vectors $f_i, \;\; i=1,2,3, \ldots, N$, which were initially used  to define the reproducing kernel, becomes an orthonormal basis for the reproducing kernel Hilbert space $\hk$ so that $\dim(\hk)=N+1$, which could be infinite.
\end{remark}

Let $\h$ be an abstract right quaternionic Hilbert space of the same dimension as $\hk$. Let $\{\Phi_i\}_{i=0}^N$ be an orthonormal basis of $\h$. Now, with the same $f_i$ of (\ref{n}), define the vectors in $\h$ by
\begin{equation}\label{eta}
\eta_x^v=\sum_{i=0}^N\Phi_i\langle f_i(x)|v\rangle_\fk;\quad x\in X,~v\in\fk.
\end{equation}
With the orthonormal basis $\{v_i\}_{i=1}^n$ of $\fk$, define the set of QCS in $\h$
\be\mathfrak{G}_\h=
\{\eta_x^i:=\eta_x^{v_i}~|~x\in X, ~i=1,2,\cdots, n\}.
\label{QCS}
\en
Define
\begin{equation}\label{uni}
W:\hk\longrightarrow\h\quad\mbox{by}\quad Wf_i=\Phi_i;\quad i=0,1,2,\cdots N.
\end{equation}
Clearly $W$ is a unitary map and the $\Phi_i$ are simply the unitary images of $f_i$
under this map. Also,
\be
\eta^v_x = W\xi^v_x,
\label{QCS-map}
\en
\begin{proposition}\label{P11}
For $x,y\in X$ and $u,v\in\fk$, we have $\langle\eta_x^u|\eta_y^v\rangle_\h=\langle u|K(x,y)v\rangle_\fk$.
\end{proposition}
\begin{proof}
Indeed,
\begin{eqnarray*}
\langle\eta_x^u|\eta_y^v\rangle_\h&=&\left\langle\sum_{i=0}^N\Phi_i\langle f_i(x)|u\rangle_\fk\mid \sum_{j=0}^N\Phi_j\langle f_j(y)|v\rangle_\fk\right\rangle_\h\\
&=&\sum_{i=0}^N\sum_{j=0}^N\overline{\langle f_i(x)|u\rangle_\fk}\langle\Phi_i|\Phi_j\rangle_\h \langle f_j(y)|v\rangle_\fk\\
&=&\sum_{i=0}^N\langle u|f_i(x)\rangle_\fk\langle f_i(y)|v\rangle_\fk
=\langle u|K(x,y)v\rangle_\fk.
\end{eqnarray*}
\end{proof}

In the physical literature, (generalized) coherent states (CS) are built on complex Hilbert
spaces, which are usually not reproducing kernel spaces, but isomorphic to them.
The construction outlined above is the most general technique for building QCS.
In most  physical situations CS are also required
to satisfy a resolution of the identity. To see this in the present context, suppose that $X$ is a
measure space with an appropriately chosen positive regular Borel measure $d\nu$,
the support of which is all of $X$ and such that
\begin{equation}\label{resol}
\sum_{i=1}^n\int_X|\eta_z^i\rangle\langle\langle\eta_z^i|d\nu(z)=I_\h,
\end{equation}
where $I_\h$ is the identity operator in $\h$. The above integral is assumed to converge weakly. This is the sense in which a family of QCS will be said to give a {\em resolution of the identity\/.}
\begin{proposition}\label{P12}
If the set of QCS $\mathfrak{G}_\h$ satisfies (\ref{resol}), then the kernel is square integrable in the sense that
$$\int_XK(x,z)K(z,y)d\nu(z)=K(x,y).$$
\end{proposition}
\begin{proof}
Suppose (\ref{resol}) holds. Then, for $u,v\in\fk$,
\begin{eqnarray*}
\sum_{i=1}^n\int_X\langle\eta_x^u|\eta_z^i\rangle_\h\langle\eta_z^i|\eta_y^v\rangle_\h d\nu(z)
=\langle\eta_x^u|\eta_y^v\rangle_\h.
\end{eqnarray*}
That is
$$\sum_{i=1}^n\int_X\langle u|K(x,z)v_i\rangle_\fk\langle v_i|K(z,y)v\rangle_\fk d\nu(z)
=\langle u|K(x,y)v\rangle_\fk.$$
Hence
$$\int_X K(x,z)\sum_{i=1}^n|v_i\rangle\langle v_i|K(z,y)d\nu(z)=K(x,y).$$
Therefore
$$\int_X K(x,z)K(z,y)d\nu(z)=K(x,y).$$
\end{proof}

\subsection{An associated positive operator valued (POV) measure}\label{subsec-POV-fcn}
On $\h$ we again have a POV-function $F:X\longrightarrow \mathcal L(\h)$, with
\begin{equation}\label{F}
F(x)=\sum_{i=1}^n|\eta_x^i\rangle\langle\eta_x^i| = W F_K (x) W^{-1}\; ,
\end{equation}
with $W$ as in (\ref{uni}).
Let $\B(X)$ be the Borel subsets of $X$. Suppose that the resolution of the identity
(\ref{resol}) holds. Let us construct a positive operator-valued set function on $\mathfrak H$  as follows.
Define the set function $a: \B(X)\longrightarrow \mathcal L(\h)$
\begin{equation}\label{povm}
a(\Delta)=\int_{\Delta}F(x)d\nu(x), \qquad \Delta \in \B(X).
\end{equation}
Clearly $a(\emptyset)=0$. From the resolution of the identity we also have $a(X)=I_\h$,
i.e.,  $a$  is  normalized.
\begin{proposition}\label{P13}
The set function $\Delta \longmapsto a(\Delta)$ is $\sigma$-additive in the weak sense. That is, if $\Delta_i\cap\Delta_j=\emptyset;~~i\not=j$, $i,j\in\N$, then
\begin{equation}\label{sigma}
a\left(\bigcup_{k=1}^{\infty}\Delta_k\right)=\sum_{k=1}^{\infty}a(\Delta_k)
\end{equation}
and the sum converges weakly.
\end{proposition}
\begin{proof}
 Let
$$S_m=\sum_{k=1}^ma(\Delta_k).$$
Since $F(x)$ is a positive self-adjoint operator, $\{S_m\}_{m=1}^{\infty}$ is a sequence of positive self-adjoint operators and
$$S_m\leq S_{m+1}\leq I_\h; \quad m\in\N.$$
Therefore, by Theorem \ref{domi}, $\{S_m\}_{m=1}^{\infty}$ converges weakly. That is
$$\lim_{m\rightarrow\infty}S_m=\sum_{k=1}^{\infty}a(\Delta_k)$$
converges weakly, which means, for $\Phi, \Psi\in\h$, the quaternionic sum $\sum_{k=1}^{\infty}\langle\Phi|a(\Delta_k)\Psi\rangle_\h$ converges. Therefore, since the elements of the sequence $\{\Delta_k\}_{k=1}^{\infty}$ are pairwise disjoint and $F(x)$ is positive, we have (\ref{sigma}).
\end{proof}

 In view of the above, we call the set function $\Delta \longmapsto a(\Delta)$ a {\em positive operator valued (POV) measure\/.}

\begin{definition}
\begin{enumerate}
\item[(i)] A POV-measure $a$ is said to be regular if for each $\Phi\in\h$, the real measure
$$\mu_{\Phi}(\Delta)=\langle\Phi|a(\Delta)\Phi\rangle$$
is regular.
\item[(ii)]$a$ is said to be bounded if $a(X)=A$ is a bounded operator. This means that the positive Borel measure $\mu_{\Phi}$ is bounded.
\item[(iii)]Let $\mu$ be a positive Borel measure on $X$. $a$ is said to be smooth with respect to the measure $\mu$ if $a(\Delta)=0$ if and only if $\mu(\Delta)=0$ for any $\Delta\in \B(X)$.
\end{enumerate}
\end{definition}
\begin{proposition}\label{P14}Suppose the measure $\nu$ in (\ref{povm}) is regular, then
\begin{enumerate}
\item[(i)]the POV-measure $a$ is regular.
\item[(ii)] $a$ is smooth with respect to $\nu$.
\end{enumerate}
\end{proposition}
\begin{proof}Let $C(\Delta)=\{C\subseteq\Delta~|~C~\mbox{is compact}\}.$
Since the measure $\nu$ is regular it satisfies
$$\nu(\Delta)=\sup_{B\in C(\Delta)}\nu(B).$$
In addition, the operator $F(x)$ is positive and bounded, therefore (i), (ii) follows immediately.
\end{proof}
\begin{proposition}\label{P15}
Let $\Phi\in\h$ and $\phi = W^{-1}\Phi \in \hk$.  Then, writing $a_K (\Delta) = W^{-1} a(\Delta ) W$
\begin{equation}\label{density}
\langle\Phi|a(\Delta)\Phi\rangle_\h= \langle \phi \mid a_K(\Delta)\phi\rangle_{\h_K} = \int_{\Delta}\|\phi(x)\|_\fk^2\; d\nu(x).
\end{equation}
\end{proposition}
The proof is clear from the isometry of $W$ (see (\ref{uni}) -- (\ref{QCS-map})).

\medskip

If we now identify $\Vert \phi (x)\Vert^2$ with a probability density, we may call $a(\Delta)$ an {\em operator of localization} in the set $\Delta$. Additionally, since the measure is normalized, from (\ref{density}) it follows that
$$ \Vert \Phi\Vert^2 = \int_X \Vert \phi (x)\Vert^2_{\mathfrak K}\;d\nu (x)= \Vert\phi\Vert^2_K . $$
In other words, the norm defined on $\h_K$ through Proposition  \ref{P2} becomes an $L^2$-norm, with the scalar product given by,
\be\label{l2-sc-prod}
  \langle\phi\mid\psi\rangle_K = \int_X\langle \phi (x)\mid \psi (x)\rangle_{\mathfrak K} \; d\nu (x).
\en

Let us denote by $L^2_{\mathfrak K} (X, d\nu)$ the right quaternionic Hilbert space of functions $f : X \longrightarrow \mathfrak K$, with scalar product defined as in (\ref{l2-sc-prod}). Then the reproducing kernel Hilbert space $\h_K$ becomes a subspace of this space. We denote by $\mathbb P_K$ the projection operator from $L^2_{\mathfrak K} (X, d\nu)$ to $\h_K$. We then have,
\begin{proposition}\label{int-ker-prop}
The reproducing kernel $K(x,y)$ is the integral kernel for the operator $\mathbb P_K$, i.e., for any $F \in L^2_{\mathfrak K} (X, d\nu)$,
\be
(\mathbb P_K F)(x) = \int_X K (x, y) F(y) \; d\nu (y).
\label{int-ker-eq}
\en
\end{proposition}
The proof follows very simply from the fact that the kernel can be written in terms of the orthonormal basis vectors $f_i$ of the subspace $\h_K \subset L^2_{\mathfrak K} (X, d\nu)$, as in (\ref{pre-K}).

\subsection{A Naimark type of an extension theorem}\label{subsec-naimark}
There is a well-known theorem, due to Naimark \cite{naimark}, on a complex Hilbert space, which says that any normalized positive operator-valued measure can be lifted to a projection-valued measure in an enlarged Hilbert space, in a certain minimal fashion. We now show how a similar extension is possible for the normalized POV-measure $\Delta \longmapsto a_K (\Delta) = W^{-1} a(\Delta ) W$.

\begin{definition}\label{D-pro}
A POV-measure $a=P$ is said to be a {\em projection valued measure}  (PV-measure) if $P(\Delta)$ is a projection operator for each $\Delta\in\B(X)$. That is, it satisfies $P(\Delta)=P(\Delta)^{\dagger}=P(\Delta)^2.$
\end{definition}

Let us now define a projection-valued measure $P(\Delta),\; \Delta \in \mathcal B(X)$ on $L^2_{\mathfrak K} (X, d\nu)$ as follows:
\begin{equation}\label{pv-m}
(P(\Delta)F)(x)=\chi_{\Delta}F(x),
\end{equation}
where $\chi_{\Delta}$ is the characteristic function of the set $\Delta$. That this defines a normalized PV-measure is easy to check. We then have the following version of the Naimark extension theorem.

\begin{theorem}\label{Naimark}
The normalized PV-measure $\Delta \longmapsto P(\Delta)$ extends the POV-measure $\Delta \longmapsto a_K(\Delta)$ in the sense of Naimark, i.e.,
\be
  a_K(\Delta) = \mathbb P_K P(\Delta)\mathbb P_K, \quad \Delta \in \mathcal B (X).
\label{ext}
\en
This extension is minimal in the sense that the set of vectors,
$$\mathcal{S}=\{P(\Delta)\phi \mid  \Delta\in \B(X), \phi \in \h_K\}$$
is dense in $L^2_{\mathfrak K} (X, d\nu)$.
\end{theorem}
\begin{proof}
Let $F \in L^2_{\mathfrak K} (X, d\nu)$  and $\phi = \mathbb P_K F \in \h_K$. Then
\beano
\langle F \mid \mathbb P_K P(\Delta )\mathbb P_K F\rangle_{L^2_{\mathfrak K} (X, d\nu)}
     & = &  \langle \phi \mid P(\Delta )\phi \rangle\\
      & = &\Vert P(\Delta )\phi\Vert^2_{L^2_{\mathfrak K} (X, d\nu)}, \quad \text{using} \;\; P(\Delta) = P(\Delta)^2 = P(\Delta)^*\\
      & = & \int_{\Delta }\Vert \phi (x)\Vert^2_{\mathfrak K}\\
       & = &\langle \phi \mid a_K(\Delta)\phi\rangle_{ L^2_{\mathfrak K} (X, d\nu)} ,\quad \text{by}\;\; (\ref{density}),
\enano
and since $a_K(\Delta)G =0$ on any vector $G$ which is in the orthogonal complement of $\h_K$ in $L^2_{\mathfrak K} (X, d\nu)$,
$$ \langle F \mid \mathbb P_K P(\Delta )\mathbb P_K F\rangle_{L^2_{\mathfrak K} (X, d\nu)} = \langle F \mid a_K (\Delta) F\rangle_{L^2_{\mathfrak K} (X, d\nu)}, $$
from which (\ref{ext}) follows by virtue of the positivity of the operator $a_K(\Delta)$.

To prove the minimality of the extension, suppose that for a fixed $F\in L^2_{\mathfrak K} (X, d\nu)$,
$$ \langle F \mid P(\Delta )\phi \rangle_{L^2_{\mathfrak K} (X, d\nu)} = \int_\Delta \langle F(x) \mid \phi (x)\rangle_{\mathfrak K} \; d\nu (x) = 0, $$
for all $\Delta \in \mathcal B (X)$ and all $\phi \in \h_K$. Thus,
$$\langle F(x) \mid \phi (x)\rangle_{\mathfrak K} =0, $$
for $\nu$-almost all $x \in X$. By virtue of (\ref{fk}), as $\phi$ runs through $\h_K$, the vectors $\phi (x)$, for fixed $x$, span $\mathfrak K$. Thus, the above equation implies that $F(x) =0$ almost everywhere, i.e., $F = 0$ as a vector in $L^2_{\mathfrak K} (X, d\nu)$,  proving the density of the set $\mathcal S$.
\end{proof}

\subsection{A first example}\label{example}
We end this section with a simple example of a reproducing kernel and its associated coherent states, which is the quaternionic equivalent of the canonical coherent states of physics (see, for example, \cite{Alibk}). These coherent states have also been reported in \cite{alibharoy,TH}.

Consider the set of quaternionic monomials,
\be
f_n (q ) = \dfrac { q^n}{\sqrt{n!}}, \quad  n = 0, 1, 2, \ldots .
\label{q-monomials}
\en
These clearly satisfy conditions (\ref{n}) and (\ref{in-1}) for the definition of a
reproducing kernel, while since now $\mathfrak K = \mathbb H$, condition (\ref{fk})
is automatically satisfied.
We thus define a reproducing kernel
\be
K( q ,  q' ) = \sum_{n=0}^\infty \frac {{\overline{ q}}^n { q'}^n}{n!} .
\label{canker}
\en
Note that in this case, $\mathcal N (\mathfrak q ) = e^{\vert \mathfrak q\vert^2}$. In the
associated Hilbert space $\h_K$ we define, following (\ref{xi})and (\ref{qcsbasis}), the
coherent states,
\be
\xi_{\overline{ q}}
=  \sum_{n=0}^\infty f_n \frac {{\overline{ q}}^n} {\sqrt{n!}} .
\label{cancs}
\en

Since the (quaternionic) Hilbert space $\h_K$ of the kernel (\ref{canker}), is generated by the  monomials $\dfrac {{ q}^n}{\sqrt{n!}}$, it is expected to consist of left (slice) regular functions, in the sense of  \cite{colsabstru,Leo,Ric}. We now show that this is indeed the case and that $\h_K$ is in fact a subspace of an $L^2$-space.

Let $q=x_0+x_1i+x_2j+x_3k$. For $\theta_1 \in [0, \frac {\pi}2), \quad \theta_2, \; \phi \in (0, 2\pi] $,
introducing the unit vector,
\begin{equation}
\widehat{\mathbf n}(\theta_1, \phi)   = (\sin\theta_1\cos\phi, \;\sin\theta_1\sin\phi,\; \cos\theta_1 ),
\label{unit-vect}
\end{equation}
the unit vector $\widehat{\mathbf n}_0$ along the $x_0$-axis and writing $x = r\cos\theta_2, \; y = r\sin\theta_2$, we may write  $\bx = (x_0, x_1, x_2, x_3) \in \mathbb R^4$ as a point in the two-dimensional plane determined by $\widehat{\mathbf n}_0$ and $\widehat{\mathbf n}(\theta_1, \phi)$,
\begin{equation}
\mathbf x = \widehat{\mathbf n}_0 x + \widehat{\mathbf n}(\theta_1, \phi) y, \qquad x,\; y\in \mathbb R .
\label{cylind-rep}
\end{equation}
The unit vectors $\widehat{\mathbf n}(\theta_1, \phi)$ for all $\theta_1 \in [0, \dfrac {\pi}2), \;\; \phi \in (0, 2\pi]$, constitute all the  points on the upper hemisphere (without the boundary rim) of the unit sphere in $\mathbb R^3$. Thus, we have the alternative coordinatization for a point in $\mathbb R^4$,
\begin{equation}
\mathbf{x} = \widehat{\mathbf n}_0 x + \widehat{\mathbf n}(\theta_1, \phi) y, \qquad \theta_1 \in [0, \frac {\pi}2 ), \;\; \phi \in [0, 2\pi), \;\; x, y \in \mathbb R .
\label{cylind-rep2}
\end{equation}
(Of course, there are points of $\mathbb R^4$ which are left out in this choice of coordinates, e.g. points determined by unit vectors along the boundary of the hemisphere, but these will turn out to be a set of measure zero in the measures that we shall shortly introduce on $\mathbb H$).
Introducing the unit imaginary quaternion,
\be
J(\theta_1, \phi ) = i\sin\theta_1\cos\phi  + j\sin\theta_1\sin\phi +  k \cos\theta_1\;,
\label{gen-un-quat}
\en
we write
\be
q = r\sin\theta_2  +  r\; J(\theta_1, \phi )\cos\theta_2 = re^{ J(\theta_1, \phi )\;\theta_2},
\label{polar-form}
\en
which is the well-known polar representation of a quaternion.  For each fixed $J(\theta_1, \phi)$, the set of all quaternions $q$, in the equation above, represents a complex plane.
The Lebesgue measure on $\mathbb R^4$ in terms of these variables is easily worked out to be (see \cite{Mu1} for details)
\begin{equation}
dx_0\;dx_1\;dx_2\; dx_3 =  \vert y\vert\;\sqrt{x^2 + y^2}\; dx\; dy\;
d\Omega (\theta_1,\phi ), \qquad d\Omega (\theta_1,\phi )= \sin\theta_1\; d\theta_1\;d\phi .
\label{leb-meas2}
\end{equation}
However, we shall be working with a different measure on $\mathbb H$. Let us define,
\be
d\nu (q, \overline{q}) = \frac 1{2\pi^2}\;e^{-r^2}\; r\sin\theta_1\; dr\;d\theta_1\; d\theta_2 \; d\phi \; .
\label{gauss-meas}
\en
Then, using the polar representation above,  we easily find that
\be
\int_{r = 0}^\infty\!\int_{\theta_1 =0}^{\frac {\pi}2\!}\int_{\theta_2 = 0}^{2\pi}\!\int_{\phi =0}^{2\pi} \overline{q}^m q^n \; d\nu (q, \overline{q}):= \int_\mathbb H \overline{q}^m q^n \; d\nu (q, \overline{q}) = n!\;\delta_{mn}\;.
\label{quat-orthog}
\en
This means that the vectors $f_n \in \h_K$, which formed an orthonormal basis for $\h_K$ are also orthonormal with respect to the measure $d\nu (q, \overline{q})$, i.e., the scalar product of $\h_K$ is an $L^2$-product and $\h_K \subset L^2_\mathbb H (\mathbb H, d\nu (q, \overline{q}))$ as a subspace. Following basically the same argument as in \cite{bargmann} we can show that $\h_K$ consists entirely of left regular functions.

Moreover, it is now easily checked that the resolution of the identity
\be
\int_\mathbb H \vert \xi_{\overline{q}}\rangle\langle \xi_{\overline{q}}\vert\; d\nu (q, \overline{q}) = I_{\h_K} \; .
\label{qcs-resolid}
\en
holds on $\h_K$. In view of their similarity with the canonical coherent states,
we may call the normalized vectors  $\vert \alpha \rangle :=
\exp{[-\frac {\vert \alpha \vert^2}2]}\;\xi_\alpha, \;\; \alpha \in \mathbb H$, the
{\em canonical quaternionic coherent states\/}. These QCS have also been reported in
\cite{Thi1}.

\section{Some examples of reproducing kernel Hilbert spaces arising from orthogonal
polynomials}\label{rep-ker-orthog-poly}
In this section we build reproducing kernels and reproducing kernel Hilbert spaces starting from real orthogonal polynomials, then with their complexified versions and quaternionic extensions. Even though some of these kernels have been considered in the literature, in different contexts \cite{alieng, EA}, we work these out here as examples of our general construction of QCS and reproducing kernel spaces.

\subsection{Reproducing kernels using Hermite polynomials}\label{subsec-real-hp}
The real Hermite polynomials $H_n (x), \; n =0,1,2, \ldots , $ are defined over $\mathbb R$
and satisfy the orthogonality relations:
\be
  \int_{\mathbb R} H_m (x) H_n (x)\; e^{-x^2}\; dx = \sqrt{\pi} 2^n n!\; \delta_{mn}\; .
\label{real-herm-orth}
\en
They are obtainable using the formula:
\be
H_n (x) = (-1)^n e^{x^2} \left(\frac d{dx}\right)^n e^{-x^2}\; .
\label{hp-form1}
\en
On the Hilbert space
$$\h_{rhp} = L^2 (\mathbb R , e^{-x^2}dx), $$
the normalized polynomials,
$$
  h_n(x) = \frac 1{[{\sqrt{\pi}\; 2^n\; n!}]^{\frac 12}}\; H_n (x), $$
where $h_0 = \pi^{-\frac 14}$, the ground state, form an orthonormal basis:
$$ \langle h_m \mid h_n\rangle_{\h_{rhp}} = \delta_{mn}\; .$$
All this, of course is standard and well-known \cite{askey}. Moreover, while for all
$x,y \in \mathbb R$
$$ \sum_{n=0}^\infty h_n (x)h_n(y) = \infty, $$
for any $0<\varepsilon<1$,
\be
  \sum_{n=0}^{\infty} \varepsilon^n h_n (x)h_n(y) < \infty.
\label{herm-finsum}
\en
Thus, the functions $f^\varepsilon_n =\varepsilon^\frac n2 h_n$ can be used to build a real
reproducing kernel, using well-known summation formulae for Hermite polynomials:
\be
K_{\varepsilon} (x, y) = \sum_{n=0}^\infty f^\varepsilon_n (x) f^\varepsilon_n (y) =
\frac 1{\sqrt{\pi(1-\varepsilon^2)}} e^{-\frac {\varepsilon^2}{1-\varepsilon^2}
\left[x^2 + y^2 -
    \frac 2\varepsilon xy\right]}
  \qquad x,y \in \mathbb R\; ,
\label{other-ker}
\en
and
$$
  K_{\varepsilon} (x,x) = \sum_{n=0}^\infty \vert f^\varepsilon_n (x)\vert^2 =
  \sum_{n=0}^\infty\frac{\varepsilon^n \vert H_n (x)\vert^2}{\sqrt{\pi}\; 2^n \;n!}
    =  \frac 1{\sqrt{\pi(1-\varepsilon^2)}}\; e^{\frac {2\varepsilon}{1+\varepsilon}x^2}\; .
$$

On $L^2 (\mathbb R, e^{-x^2}\;dx)$ define the unbounded operator $A$, by its action on the basis
vectors $h_n$:
\be
  Ah_n = \varepsilon^{-n} h_n , \quad n = 0,1,2, \ldots , \quad \Longrightarrow \quad
     A = \sum_{n=0}^\infty \varepsilon^{-n}\vert h_n\rangle\langle h_n \vert\; .
\label{unbdd-op}
\en
This operator is unbounded, but (as can be seen easily) closed. Its inverse is bounded, trace
class, with $\text{Tr}[A^{-1}] = \dfrac 1{1- \varepsilon}$. Its domain, $\mathcal D( A)$,
consists
of all vectors $f = \sum_{n=0}^\infty c_n h_n \in L^2 (\mathbb R, e^{-x^2}\;dx)$, such that
$\sum_{n=0}^\infty \varepsilon^{-2n}\vert c_n\vert^2 < \infty$\;. On $\mathcal D (A)$
define a new scalar product
\be
 \langle f \mid g \rangle_{\h_{\varepsilon}} = \int_{\mathbb R}\overline{f(x)} (Ag)(x)\;
     e^{-x^2}\; dx\; .
\label{new-scprod}
\en
and denote by $\h_\varepsilon$ its completion with respect to this norm. Clearly, as  sets
$\h_{\varepsilon} \subset L^2 (\mathbb R, e^{-x^2}\;dx)$ and the vectors
$f_n^{\varepsilon} = A^{-\frac 12}h_n = \varepsilon^{\frac n2}h_n\; , \;\; n=0,1,2, \ldots,$
form an orthonormal basis of
$\h_{\varepsilon}$. Moreover, $\h_{\varepsilon}$ is a
reproducing kernel Hilbert space, with reproducing kernel $K_{\varepsilon} (x,y)$ in (\ref{other-ker}).
However, because of the nature of the scalar product (\ref{new-scprod}), it is not an
$L^2$-space. (This space has been studied in detail, in connection with Berezin-Toeplitz
quantization in \cite{alieng}.)  However, if we extend the Hermite polynomials to the complex plane,
i.e., if we consider them to be functions $H_n (z)$ and $f^\varepsilon_n (z)$ of a complex variable, then there the resulting reproducing kernel

\bea
  K_{\varepsilon} (z, \overline{w}) &= & \sum_{n=0}^\infty f^\varepsilon_n (z)
  \overline{f^\varepsilon_n (w)} =  \sum_{n=0}^\infty
  \frac {\varepsilon^n}{\sqrt{\pi}\; 2^n\; n!}\; H_n (z)H_n(\overline{w})\nonumber\\
    & = &
\frac 1{\sqrt{\pi(1-\varepsilon^2)}} e^{-\frac {\varepsilon^2}{1-\varepsilon^2}
\left[z^2 + \overline{w}^2 -
    \frac 2\varepsilon z\overline{w}\right]}
  \qquad z,w \in \mathbb C\; ,
\label{other-comp-ker}
\ena
is square-integrable in the sense of (\ref{resol}) and Proposition \ref{P12}. Indeed, we have \cite{karp,karp2,eindmey},
\be
\int_{\mathbb C}  \overline{f^\varepsilon_m (z)}  f^\varepsilon_n (z)\; d\nu_\varepsilon (z, \overline{z}) = \delta_{m,n},
\label{com-orthog}
\en
where
$$  d\nu_\varepsilon (z, \overline{z}) = \frac {\sqrt{1 -\varepsilon^2}}{2\varepsilon}\; \exp \left[-2\varepsilon\{ \frac {x^2}{1 + \varepsilon} +  \frac {y^2}{1 - \varepsilon}\}\right]\; dx\;dy\; , \qquad z = x+ iy, $$
so that,
\be
  \int_{\mathbb C}  K_{\varepsilon} (z, \overline{z}') K_{\varepsilon} (z', \overline{w})\; d\nu_\varepsilon (z', \overline{z}') =  K_{\varepsilon} (z, \overline{w})\; .
\label{comp-sq-int}
\en
Thus the functions $f^\varepsilon_n (z), \; n =0,1,2, \ldots, $ span a reproducing kernel subspace, $\h^{\text{hol}}_\varepsilon$, of the Hilbert space $L^2 (\mathbb C, d\nu_\varepsilon (z, \overline{z}))$, consisting of holomorphic functions and the restrictions of the functions in $\h^{\text{hol}}_\varepsilon$ to the real axis then consist exactly of all the functions in the Hilbert space $\h_\varepsilon$ above, with the real kernel (\ref{other-ker}).

  We now extend the  $f^\varepsilon_n, \; n =0,1,2, \ldots,$ to functions of a quaternionic variable:
\be
f^\varepsilon_n (q) =   \sqrt{\frac {\varepsilon^n}{\sqrt{\pi}\; 2^n\; n!}}\; H_n (q),
\label{quat-herm-pol}
\en
It turns out that these functions again define a quaternionic reproducing kernel and a Hilbert space. Indeed, it is clear that the functions $f^\varepsilon_n (q), \; n =0,1,2, \ldots,$ satisfy the conditions (\ref{n}) -- (\ref{fk}) with
$$
\mathcal N (q) =  \sum_{n=0}^\infty \vert f^\varepsilon_n (q)\vert^2 =  \sum_{n=0}^\infty
  \frac {\varepsilon^n}{\sqrt{\pi}\; 2^n\; n!}\; \vert H_n (q)\vert^2 =
  \frac 1{\sqrt{\pi(1-\varepsilon^2)}} e^{-\frac {\varepsilon^2}{1-\varepsilon^2}
\left[q^2 + \overline{q}^2 -
    \frac 2\varepsilon \vert q\vert^2\right]}.
$$
as follows immediately from (\ref{other-comp-ker}). Thus, the kernel (\ref{other-comp-ker}) extends to a quaternionic reproducing kernel
\be
K_{\varepsilon} (q, \overline{q}') =  \sum_{n=0}^\infty f^\varepsilon_n (q)
  \overline{f^\varepsilon_n (q')} =  \sum_{n=0}^\infty
  \frac {\varepsilon^n}{\sqrt{\pi}\; 2^n\; n!}\; H_n (q)H_n(\overline{q}')
\label{quat-herm-ker}
\en
However this time, in view of the general non-commutativity of $q$ and $q'$, we are unable to get a closed form expression for this kernel, as was possible for the complex and real cases above.

With the parametrization for a quaternion given in (\ref{polar-coords}) and writing $x = r\cos\theta_2, \; y= r\sin\theta_2$, let us introduce the measure,
\be
  d\nu_\epsilon (q , \overline{q}) = \frac {\sqrt{1 -\varepsilon^2}}{4\pi\varepsilon}\; \exp \left[-2\varepsilon\{ \frac {x^2}{1 + \varepsilon} +  \frac {y^2}{1 - \varepsilon}\}\right]\; dx\;dy\;\sin\theta_1\; d\theta_1\; d\phi\; ,
\label{quat-meas}
\en
with $x, y \in \mathbb R, \;\; \theta_1 \in [0, \frac {\pi}2), \;\; \phi \in (0, 2\pi]\; .$
Then, using the polar representation (\ref{polar-form}) of quaternions and the orthogonality relation (\ref{com-orthog}), we easily find the orthogonality condition for the quaternionic polynomials (see also \cite{Thi1}),
\be
   \int_\mathbb H f^\varepsilon_m (\overline{q}) f^\varepsilon_n (q) \; d\nu_\epsilon (q , \overline{q}) = \delta_{mn}\; ,
\label{quat-orthog2}
\en
i.e., the reproducing kernel $K_{\varepsilon} (q, \overline{q}')$ is square-integrable and the
associated Hilbert space $\h_\varepsilon^{\text{reg}}$ is a subspace of
$L^2_\mathbb H (\mathbb H, d\nu_\epsilon (q , \overline{q})$, consisting again of left
regular functions.

 We end this section by recalling that we constructed here, starting with the normalized real
Hermite polynomials, $h_n (x)$, a first reproducing kernel Hilbert space, $\h_\varepsilon$,
for which the polynomials $f^\varepsilon_n (x) = \varepsilon^{\frac n2}h_n (x)$ formed an
orthonormal basis. However, the associated reproducing kernel, $K_\varepsilon (x,y)$,
was not square-integrable.
We then  extended these polynomials to the complex plane, to obtain polynomials,
$f^\varepsilon_n (z)$, in a complex
variable, which then formed an orthonormal basis in a reproducing kernel
Hilbert space, $\h_\varepsilon^{\text{hol}}$, consisting of holomorphic functions. Moreover
the associated kernel, $K_\varepsilon (z, \overline{w})$, which was just the complex extension of the real kernel $K_\varepsilon (x,y)$, was square-integrable and the
Hilbert space $\h_\varepsilon^{\text{hol}}$ was  a subspace of
a complex $L^2$-space. Finally, we further extended the polynomials $f^\varepsilon_n (z)$ to
be functions, $f^\varepsilon_n (q)$, of a quaternionic variable. These also formed an orthonormal
basis in a right quaternionic Hilbert space, $\h_\varepsilon^{\text{reg}}$, of left regular
functions and once again the
associated kernel, $K_\varepsilon (q, \overline{q}')$ was square integrable and the Hilbert space
$\h_\varepsilon^{\text{reg}}$ a subspace of a quaternionic $L^2$-space.

\subsection{Reproducing kernels using Laguerre polynomials}
In this section we briefly show how the same sort of analysis may be done using the real Laguerre
polynomials, again successively obtaining three reproducing kernel Hilbert spaces --
the first consisting of functions of a real variable,
the second of a complex variable and the third a quaternionic Hilbert space of a
quaternionic variable (see, also \cite{EA,karp,karp2}).

The generalized real Laguerre polynomials are defined for any $\alpha > -1$ by
\be
 L^\alpha_n(x) = \sum_{k=0}^n\frac {\Gamma (n+ \alpha +1)}{\Gamma (k + \alpha +1)
              \; \Gamma (n-k+1)\; k!}\; (-x)^k\; , \qquad n = 0, 1,2, \ldots ,
\label{realaguer}
\en
which can also be obtained using
$$
  L^\alpha_n(x) = e^x x^{-\alpha}\; \frac 1{n!}\; \frac d{dx^n} (e^{-x} x^{n+\alpha})\; .
$$
They satisfy the orthogonality relations,
\be
  \int_0^\infty L_j^\alpha (x) \; L_k^\alpha (x)\; x^\alpha e^{-x}\; dx =
      \frac {\Gamma (k+\alpha +1)}{k!}\; \delta_{jk} \; .
\label{lagorthog}
\en
Once again, while
$$\sum_{n=0}^\infty\vert L^\alpha_n (x)\vert^2 = \infty, $$
one has, for any $\varepsilon \in (0, 1)$ (see, for example, \cite{karp,karp2}),
\be
  K^\alpha_\varepsilon (x,y) = \sum_{n=0}^\infty \varepsilon^n \widehat{L}^\alpha_n (x)\; \widehat{L}^\alpha_n (y) = \frac 1{1 -\varepsilon}\; \exp\left[ - \frac {\varepsilon ( x + y)}{1-\varepsilon}\right]\; (\varepsilon xy)^{-\frac {\alpha}2}\; I_\alpha \left( \frac {2\sqrt{\varepsilon xy}}{1 - \varepsilon}\right)\; ,
\label{realagker}
\en
where $\widehat{L}^\alpha_n$ are the normalized polynomials,
$$\widehat{L}^\alpha_n (x) = \left[\frac {k!}{\Gamma (k+\alpha +1)}\right]^{\frac 12}\; L^\alpha_n (x)\; , $$
and $I_\alpha$ is the modified Bessel function of the first kind.

As is well known (see, for example, \cite{szego}), the normalized Laguerre polynomials form an orthonormal basis for the Hilbert space $\h_{rlg} = L^2 (\mathbb R^+, x^\alpha\; e^{-x}\; dx )$. On the other hand, in view of (\ref{realagker}), the polynomials
\be
   f^{\alpha, \varepsilon}_n (x) = \varepsilon^{\frac n2}\widehat{L}^\alpha_n (x) =  \left[\frac {\varepsilon^n k!}{\Gamma (k+\alpha +1)}\right]^{\frac 12}\; L^\alpha_n (x)\; \qquad n =0,1,2, \ldots ,
\label{epslag}
\en
form an orthonormal basis for a reproducing kernel Hilbert space $\h^\alpha_\varepsilon$, with (\ref{realagker}) as the reproducing kernel. As sets, $\h^\alpha_\varepsilon \subset \h_{rlg}$ and as in the case with Hermite polynomials, the elements of $\h^\alpha_\varepsilon$ form the domain of an unbounded operator $A$ on $\h_{rlg}$, which acts on the  basis vectors in the manner $A \widehat{L}^\alpha_n = \varepsilon^{-n}\widehat{L}^\alpha_n$. Once again while $\h^\alpha_\varepsilon$ is not an $L^2$-space, the polynomials $f^{\alpha, \varepsilon}_n$, when extended to the complex plane, satisfy the orthogonality relations \cite{karp,eindmey},
\be
  \int_\mathbb C \overline{f^{\alpha, \varepsilon}_m (z)}f^{\alpha, \varepsilon}_n (z)\;
      d\nu^\alpha_\varepsilon (z, \overline{z}) = \delta_{mn},
\label{complagorth}
\en
where
$$d\nu^\alpha_\varepsilon (z, \overline{z}) = \frac {2c\;\varepsilon^{\frac {\alpha}2}}\pi\;\exp[2cx]\; \vert z\vert^\alpha\; K_\alpha \left( \frac {2\sqrt{\varepsilon}\;\vert z\vert}{1-\varepsilon}\right)\; dx\; dy, \qquad c = \frac \varepsilon{ 1 - \varepsilon}, \quad z = x + iy\; , $$
and $K_\alpha$ is the modified Bessel function of the second kind. Thus these polynomials form a closed subspace $\h^{\text{hol}}_{\alpha, \varepsilon}$, consisting of holomorphic functions, of the Hilbert space $L^2 (\mathbb C ,  d\nu^\alpha_\varepsilon (z, \overline{z}))$. Moreover, the restriction of the functions in this subspace to the real line constitute all the functions in Hilbert space $\h^\alpha_\varepsilon$, with the real reproducing kernel (\ref{realagker}).  The subspace $\h^{\text{hol}}_{\alpha, \varepsilon}$ is again a reproducing kernel Hilbert space with the kernel
\be
K^\alpha_\varepsilon (z,\overline{w}) = \sum_{n=0}^\infty f^{\alpha, \varepsilon}_n (z)\; \overline{f^{\alpha, \varepsilon}_n (w)} = \frac 1{1 -\varepsilon}\; \exp\left[ - \frac {\varepsilon ( z + \overline{w})}{1-\varepsilon}\right]\; (\varepsilon z\overline{w})^{-\frac {\alpha}2}\; I_\alpha \left( \frac {2\sqrt{\varepsilon z\overline{w}}}{1 - \varepsilon}\right)\; ,
\label{complagker}
\en
which is the extension of the kernel (\ref{realagker}) to $\mathbb C$. This kernel is square-integrable,
\be
  \int_{\mathbb C}  K_{\alpha, \varepsilon} (z, \overline{z}') K_{\alpha,\varepsilon} (z', \overline{w})\; d\nu^\alpha_\varepsilon (z', \overline{z}') =  K_{\alpha, \varepsilon} (z, \overline{w})\; .
\label{comp-lag-sq-int}
\en

It is now straightforward to build the corresponding quaternionic Laguerre polynomials and their orthogonality relations. We simply extend the polynomials $f^{\alpha, \varepsilon}_n (z)$ to $f^{\alpha, \varepsilon}_n (q), \; q \in \mathbb H$. Then, following (\ref{quat-meas}) and (\ref{quat-orthog2}) we see that these polynomials satisfy the orthogonality relation,
\be
  \int_\mathbb C \overline{f^{\alpha, \varepsilon}_m (q)}f^{\alpha, \varepsilon}_n (q)\;
      d\nu^\alpha_\varepsilon (q, \overline{q}) = \delta_{mn},
\label{quatlagorth}
\en
where again, with the same parametrization of the quaternions as in (\ref{polar-coords}), the measure $d\nu^\alpha_\varepsilon (q, \overline{q})$ is defined as
$$d\nu^\alpha_\varepsilon (q, \overline{q}) = \frac {c\;\varepsilon^{\frac {\alpha}2}}{\pi^2}\;\exp[2cx]\; \vert z\vert^\alpha\; K_\alpha \left( \frac {2\sqrt{\varepsilon}\;\vert z\vert}{1-\varepsilon}\right)\; dx\; dy\;\sin\theta_1\;d\theta_1\; d\phi \; ,  $$
(with $c = \dfrac \varepsilon{ 1 - \varepsilon}, \; z = x + iy$). These polynomials give an orthonormal basis for a quaternionic reproducing kernel Hilbert space, $\h^{\text{reg}}_{\alpha, \varepsilon}$, which is a subspace of the Hilbert space $L^2 (\mathbb H, d\nu^\alpha_\varepsilon (q, \overline{q})$. The associated reproducing kernel is
\be
K^\alpha_\varepsilon (q, \overline{q}') = \sum_{n=0}^\infty f^{\alpha, \varepsilon}_n (q)\overline{f^{\alpha, \varepsilon}_n (q')}  = \sum_{n=0}^\infty \frac{\varepsilon^n n!}{\Gamma (n+\alpha + 1)} L^\alpha_n (q) L^\alpha_n (\overline{q}')\; ,
\label{quatlagker}
\en
which is square-integrable with respect to the measure $d\nu^\alpha_\varepsilon (q, \overline{q})$:
\be
\int_{\mathbb H} K^\alpha_\varepsilon (q, \overline{q}'') K^\alpha_\varepsilon (q'', \overline{q}')\; d\nu^\alpha_\varepsilon (q'', \overline{q}'') = K^\alpha_\varepsilon (q, \overline{q}')\; .
\label{quatlag-sq-int}
\en

An entirely analogous construction could be carried out to obtain quaternionic Jacobi polynomials, using the fact that the real Jacobi polynomials have complex orthogonal extensions as well \cite{karp}.

\subsection{A Naimark extension}\label{subsec-naim-ext}
Following the construction outlined in Section \ref{subsec-naimark}, we can now show the existence of a quaternionic POV measure and a Naimark type of extension to a quaternionic PV measure. We do this only for the case of the Laguerre polynomials, a similar construction being possible for the other cases as well.

The vectors
\be
   \xi^{\alpha, \varepsilon}_{\overline{q}} = \sum_{n=0}^\infty f^{\alpha, \varepsilon}_n\; \overline{f^{\alpha, \varepsilon}_n (q)} \in \h^{\text{reg}}_{\alpha, \varepsilon}
\label{quatlagcs}
\en
form a family of (non-normalized) quaternionic coherent states, which satisfy the resolution of the identity
\be
   \int_{\mathbb H}\mid \xi^{\alpha, \varepsilon}_{\overline{q}}\rangle\langle \xi^{\alpha, \varepsilon}_{\overline{q}}\mid\; d\nu^\alpha_\varepsilon (q, \overline{q}) = I_{\text{reg}} ,
\en
on $\h^{\text{reg}}_{\alpha, \varepsilon}$. The operators
\be
  a^{\alpha, \varepsilon} (\Delta ) = \int_{\Delta}\mid \xi^{\alpha, \varepsilon}_{\overline{q}}\rangle\langle \xi^{\alpha, \varepsilon}_{\overline{q}}\mid\; d\nu^\alpha_\varepsilon (q, \overline{q}),
\label{lagpovmeas}
\en
for Borel sets $\Delta \subset \mathbb H$, form a normalized  POV-measure on $\h^{\text{reg}}_{\alpha, \varepsilon}$. The Naimark extension to a PV-measure, $P^{\alpha, \varepsilon} (\Delta )$, is then given by (see (\ref{pv-m}))
\be
   (P^{\alpha, \varepsilon} (\Delta )F)(q) = \chi_\Delta (q) F(q), \qquad F \in  L^2 (\mathbb H, d\nu^\alpha_\varepsilon (q, \overline{q})\; .
\label{lagpvmeas}
\en

\section{Conclusion}
As mentioned in the Introduction, a general theory of reproducing kernels and reproducing kernel Hilbert spaces has been developed in this paper on right quaternionic Hilbert spaces, by analogy with their complex counterparts.  Using the reproducing kernels, a class of generalized CS, POV-measures and a Naimark extension theorem have been studied. The theory has been illustrated with quaternionic Hermite and Laguerre polynomials.
The construction presented here can be adapted to other classes of orthogonal polynomials as long as the real orthogonal polynomials have complex extensions and thereby quaternionic extensions. For example, the two indexed quaternionic Hermite polynomials
$$H_{n,m}(q,\oqu)=n!m!\sum_{j=0}^{\text{min}\{n,m\}}\frac{(\oqu)^{n-j}}{(n-j)!}\frac{q^{m-j}}{(m-j)!}$$
satisfy the orthogonality relation
$$\int_\quat \overline{H_{n,m}(q,\oqu)}~H_{l,k}(q,\oqu)\frac{1}{\pi}e^{-|z|^2}d^2zd\omega(u_q)=n!m!\delta_{nl}\delta_{mk}.$$
For fixed $n$ or $m$ we have
$$\sum_{m=0}^{\infty}\frac{|H_{n,m}(q,\overline{q})|^2}{n!m!}<\infty\quad\mbox{and}\quad
\sum_{n=0}^{\infty}\frac{|H_{n,m}(q,\overline{q})|^2}{n!m!}<\infty$$
respectively. Therefore, for each fixed $n$ or for each fixed $m$, a parallel reproducing kernel theory as for the case of $H_n(q)$ can be obtained (see \cite{Thi1}). It should be possible to develop a quaternionic version of coherent state quantization or integral quantization (see, for example, \cite{Alibk}) using the theory developed here. We plan to study this problem in future publications.

\end{document}